\documentclass[a4paper,UKenglish,cleveref, autoref, thm-restate]{lipics-v2021}

\hideLIPIcs  

\usepackage[square,numbers,comma]{natbib}

\DeclareMathOperator{\ldeg}{ldeg}
\DeclareMathOperator{\hdeg}{hdeg}
\DeclareMathOperator{\R}{\mathbb{R}}

\usepackage{float}

\let\oldexists\exists
\let\exists\relax
\DeclareMathOperator{\exists}{\oldexists}

\let\oldforall\forall
\let\forall\relax
\DeclareMathOperator{\forall}{\oldforall}

\newcommand{\restr}[2]{%
    \left.\kern-\nulldelimiterspace
    #1
    \vphantom{\big|}
    \right|_{#2}%
}

\usepackage{csquotes} 
\MakeOuterQuote{"}

\usepackage{mathtools}               
\usepackage{xfrac}                   
\usepackage{tikz-cd}                 

\usepackage{footmisc}
\usepackage[capitalise]{cleveref} 

\usepackage{standalone}

\usepackage{natbib}

\theoremstyle{plain}

\usepackage{thmtools, thm-restate}

\AddToHook{env/restLemma/begin}{\crefalias{theorem}{restLemma}}

\crefname{restLemma}{{Lemma}}{{Lemma}s}
\Crefname{restLemma}{Lemma}{Lemmas}

\crefname{theorem}{Theorem}{Theorems}
\Crefname{theorem}{Theorem}{Theorems}
\AddToHook{env/definition/begin}{\crefalias{theorem}{definition}}
\crefname{definition}{{Definition}}{{Definition}s}
\Crefname{definition}{Definition}{Definitions}

\AddToHook{env/notation/begin}{\crefalias{theorem}{notation}}
\crefname{notation}{{Notation}}{{Notation}s}
\Crefname{notation}{Notation}{Notations}

\AddToHook{env/lemma/begin}{\crefalias{theorem}{lemma}}
\crefname{lemma}{{Lemma}}{{Lemma}s}
\Crefname{lemma}{Lemma}{Lemmas}

\AddToHook{env/corollary/begin}{\crefalias{theorem}{corollary}}
\crefname{corollary}{{Corollary}}{{Corollarie}s}
\Crefname{corollary}{Corollary}{Corollaries}

\AddToHook{env/example/begin}{\crefalias{theorem}{example}}
\crefname{example}{{Example}}{{Example}s}
\Crefname{example}{Example}{Examples}

\AddToHook{env/fact/begin}{\crefalias{theorem}{fact}}
\crefname{fact}{{Fact}}{{Fact}s}
\Crefname{fact}{Fact}{Facts}

\AddToHook{env/remark/begin}{\crefalias{theorem}{remark}}
\crefname{remark}{{Remark}}{{Remark}s}
\Crefname{remark}{Remark}{Remarks}

\usepackage{algorithm}
\usepackage[noend]{algpseudocode}

\tikzset{
    D/.style={draw, circle, inner sep=0pt,
    minimum size=4pt, fill},
    S/.style={draw,circle,inner sep=0pt,
    minimum size=3pt,fill},
    B/.style={draw=none, circle, inner sep=0pt,
    minimum size=4pt, fill=blue},
    R/.style={draw=none, circle, inner sep=0pt,
    minimum size=4pt, fill=red},
    G/.style={draw=none, circle, inner sep=0pt,
    minimum size=4pt, fill=gray}
}

\usepackage{microtype} 

\title{Which Vertical Graphs are Non VPHT Reconstructible?}
\titlerunning{Which Vertical Graphs are Non VPHT Reconstructible?}

\author{Jette Gutzeit}
{Department of Mathematics, University of Zurich, Switzerland}
{jette.gutzeit@uzh.ch}
{}
{}

\author{Kalani Kistler}
{Department of Mathematics, ETH Zürich, Switzerland}
{kkistler@student.ethz.ch}
{}
{}

\author{Tim Ophelders}
{Department of Information and Computing Sciences, Utrecht University, Netherlands \and
Department of Mathematics and Computer Science, TU Eindhoven, Netherlands}
{t.a.e.ophelders@tue.nl}
{}
{}

\author{Anna Schenfisch}
{Department of Mathematics, KTH Royal Institute of Technology, Sweden}
{schenf@kth.se}
{}
{}
\authorrunning{J. Gutzeit, K. Kistler, T. Ophelders, A. Schenfisch}
\ArticleNo{100}

\newcommand{\commentHack}[1]{}

\Copyright{Jette Gutzeit, Kalani Kistler, Tim Ophelders, Anna Schenfisch} 

\ccsdesc[500]{Theory of computation~Computational geometry} 

\keywords{topological data analysis, persistent homology, graph reconstruction, verbose persistence diagrams} 

\category{} 

\relatedversion{} 

\EventEditors{John Q. Open and Joan R. Access}
\EventNoEds{2}
\EventLongTitle{42nd Conference on Very Important Topics (CVIT 2016)}
\EventShortTitle{CVIT 2016}
\EventAcronym{CVIT}
\EventYear{2016}
\EventDate{December 24--27, 2016}
\EventLocation{Little Whinging, United Kingdom}
\EventLogo{}
\SeriesVolume{42}
\ArticleNo{23}

\begin{document}

\nolinenumbers
\maketitle
\begin{abstract}
    The verbose persistent homology transform (VPHT)
    is a topological summary of shapes in Euclidean space. Assuming general position, the VPHT is injective, meaning shapes can be reconstructed using only the VPHT.
    In this work, we investigate cases in which the VPHT is not
    injective, focusing on a simple setting of degeneracy;
    graphs whose vertices are all collinear. We identify both necessary properties and sufficient
    properties for non-reconstructibility of such graphs, bringing us closer to
    a complete classification.
\end{abstract}

\section{Introduction}\label{prelim}
Persistence diagrams are a topological tool for data analysis. The persistent homology transform (PHT) is the collection of these diagrams corresponding to height filtrations of some shape in every possible direction. First introduced in \cite{turner2014persistent}, the PHT was shown to be injective
for simplicial complexes embedded in general position in dimensions $d=2,3$, meaning that we can use it to reconstruct them up to subdivision.
Injectivity was later shown for general
$d$ in various settings~\cite{ghrist2018persistent, ji2024injectivity,belton_reconstructing_2020, curry2022many,
fasy_faithful_2024, betthauser2018topological, micka2020searching,onus2024shoving, turner2024extended}.
In this work, we take a rather opposite approach, and instead investigate which
simplicial complexes are \emph{not} reconstructible by homology transforms. 
In the following, we consider the more discriminative \emph{verbose} variant of persistence
diagrams, which give rise to the \emph{verbose} PHT (VPHT), described in \cref{sec:preliminaries}.
While~\cite{ji2024injectivity} identifies families of functions not reconstructible by the related Euler characteristic transform, we are not aware of other work specifically focused on non-reconstructibility of simplicial complexes using the (V)PHT.
It is easy to find simplicial complexes that are non
PHT-reconstructible, but it is less straightforward to do so for the more
discriminative VPHT. Furthermore, as we already know the VPHT is injective
for simplicial complexes whose vertices are in general
position~\cite{fasy_faithful_2024}, we turn to a simple degenerate setting as
our starting place, namely, \emph{vertical graphs}, whose vertices are all
collinear.
Our main results identify families and characterizations of non
VPHT-reconstructible vertical graphs. We additionally observe that our results
inform the related question of choosing directions for more general
application settings; directions that see vertices and edges of a general graph in the
same order as a non-reconstructible vertical graph will not be able to
effectively determine the edge set on their own.


\section{Preliminaries}
\label{sec:preliminaries}
We first discuss the VPHT and related concepts. We consider graphs $G
= (V, E)$ in~$\R^2$, but note that all definitions generalize to simplicial
complexes in $\R^d$ (see, e.g.,~\cite{fasy_faithful_2024}). Given a unit vector
$s$, the height of a vertex $v$ with respect to~$s$ is the dot product $h_s(v) =
s \cdot v$. We extend this to edges as $h_s(v,w) = \max\{h_s(v), h_s(w)\}$.
Then, the \emph{lower-star filtration of $G$ with respect to $s$} is the nested
inclusion of subgraphs of~$G$ with increasing maximum height.

Associated to every lower-star filtration are compatible total orderings of $V
\cup E$, where, for~$\sigma_i, \sigma_j \in V \cup E$, the relation $\sigma_i
\prec \sigma_j$ implies $h_s(\sigma_i) \leq h_s(\sigma_j)$, i.e., total
orders that respect the height relation.
From such orders, we build \emph{compatible index filtrations},
or nested inclusions of subgraphs of $G$ with increasing maximum value in the total order. See~\cref{fig:VPHTfigure}.

Finally, we introduce \emph{persistence diagrams}, multisets of points in $(\R
\cup\{\infty\})^2$, where  a point $(a,b)$ represents the lifespan of a connected component or cycle. Each feature is \emph{born} at its first appearance, $a$ in the filtration. A connected component \emph{dies} at height $b$ when it merges with an older component. If this never occurs, $b=\infty$. Cycles always live to infinity. Persistence diagrams that arise from a lower-star filtration are examples of (standard) \emph{concise diagrams}, and persistence diagrams that arise from compatible index filtrations are called \emph{verbose diagrams}. See~\cref{fig:VPHTfigure}, right.

 Verbose diagrams differ from their concise counterparts only in that they have additional points on the diagonal. 
 Verbose persistence is well-defined, see, e.g., \cite[Appendix B]{fasy_faithful_2024}. 
While a full discussion is out of the scope of this writeup, we note that
verbose diagrams are not new;~\cite{usher2016persistent} establishes verbose
diagrams as arising from filtered chain complexes. This perspective is also
taken by, e.g., \cite{chacholski2023decomposing, memoli2022ephemeral}. 

We utilize the following important observation about verbose persistence diagrams.
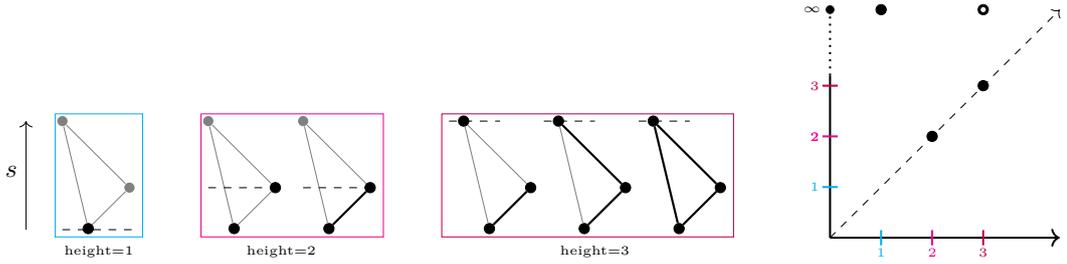
\begin{figure}
    \centering
    \begin{tikzpicture}
    \draw [->](-0.5,-1.5)--(-0.5,0);
    \node at (-0.7,-0.7){$s$};

    \node (a0) [G]  at (0,0){};
    \node (a1)[G, below right of=a0, node distance=13mm] {};
    \node (a2) [D, below left of=a1, node distance=8mm] {};
    \draw [gray] (a0)--(a1)--(a2)--(a0);
    \draw[dashed] (0,-1.5)--(1,-1.5);
    \draw[cyan] (-0.1,0.1)--(1.1,0.1)--(1.1,-1.6)--(-0.1,-1.6)--(-0.1,0.1);
    \node  at (0.5,-1.8){\tiny height=1};

    \node (b0) [G]  at (2,0){};
    \node (b1)[D, below right of=b0, node distance=13mm] {};
    \node (b2) [D, below left of=b1, node distance=8mm] {};
    \draw [gray] (b0)--(b1)--(b2)--(b0);
    \draw[dashed] (2,-0.92)--(3,-0.92);
    \node (c0) [G]  at (3.3,0){};
    \node (c1)[D, below right of=c0, node distance=13mm] {};
    \node (c2) [D, below left of=c1, node distance=8mm] {};
    \draw [gray] (c2)--(c0)--(c1);
    \draw [line width = .3mm] (c1)--(c2);
    \draw[dashed] (3.3,-0.92)--(4.3,-0.92);
    \draw[magenta] (1.9,0.1)--(4.4,0.1)--(4.4,-1.6)--(1.9,-1.6)--(1.9,0.1);
    \node at (3,-1.8){\tiny height=2};

    \node (d0) [D]  at (5.5,0){};
    \node (d1)[D, below right of=d0, node distance=13mm] {};
    \node (d2) [D, below left of=d1, node distance=8mm] {};
    \draw [gray] (d0)--(d1)--(d2)--(d0);
    \draw[line width = .3mm] (d1)--(d2);
    \draw[dashed] (5.3,0)--(6,0);
    \node (e0) [D]  at (6.8,0){};
    \node (e1)[D, below right of=e0, node distance=13mm] {};
    \node (e2) [D, below left of=e1, node distance=8mm] {};
    \draw [gray] (e0)--(e1)--(e2)--(e0);
    \draw[line width = .3mm]  (e0)--(e1)--(e2);
    \draw[dashed] (6.6,0)--(7.3,0);
    \node (f0) [D]  at (8.1,0){};
    \node (f1)[D, below right of=f0, node distance=13mm] {};
    \node (f2) [D, below left of=f1, node distance=8mm] {};
    \draw [line width = .3mm] (f0)--(f1)--(f2)--(f0);
    \draw[dashed] (7.9,0)--(8.6,0);
    \draw[purple] (5.2,0.1)--(9.2,0.1)--(9.2,-1.6)--(5.2,-1.6)--(5.2,0.1);
    \node at (7.3,-1.8){\tiny height=3};
\end{tikzpicture}
\hspace{0.6cm}
\begin{tikzpicture}[scale=0.7]
    \draw[thick](0,0)--(0,3.2);
    \draw[thick,dotted](0,3.2)--(0,4.5);
    \node[S] at (0,4.5){};
    \draw[->,thick](0,0)--(4.5,0);
    \node [cyan] at (1,-0.3){\tiny 1};
    \draw[cyan,thick] (1,-0.15)--(1,0.15);
    \node [magenta] at (2,-0.3){\tiny 2};
    \draw[magenta,thick] (2,-0.15)--(2,0.15);
    \node [purple] at (3,-0.3){\tiny 3};
    \draw[purple,thick] (3,-0.15)--(3,0.15);
    \node [cyan] at (-0.3,1){\tiny 1};
    \draw[cyan,thick] (-0.15,1)--(0.15,1);
    \node [magenta] at (-0.3,2){\tiny 2};
    \node [magenta] at (-0.3,2){\tiny 2};
    \draw[magenta,thick] (-0.15,2)--(0.15,2);
    \node [purple] at (-0.3,3){\tiny 3};
    \draw[purple,thick] (-0.15,3)--(0.15,3);
    \node at (-0.35,4.5){\tiny $\infty$};
    \draw[dashed,->](0,0)--(4.5,4.5);
    \node[D] at (1,4.5){};
    \node[D] at (2,2){};
    \node[D] at (3,3){};
    \node[D] at (3,4.5) {};
    \node[circle, draw=white,draw=none, circle, inner sep=0pt,
    minimum size=2pt, fill=white] at (3,4.5){};
\end{tikzpicture}
    \caption{The lower-star filtration of this graph with respect to the vertical up direction $s$ has a step for each vertex height, adding; (1) the bottom vertex (which creates a
    connected component), (2) the middle vertex and first edge (no effect on
    homology), and (3) the top vertex and remaining edges (creating a cycle). A corresponding index filtration is shown in the boxes, adding vertices and edges individually. This causes instantaneous births/deaths of components,
    which appear as on-diagonal points in the verbose diagram (right). The circle denotes the cycle born at height (3) and the black dots denote connected components.}
    \label{fig:VPHTfigure}
\end{figure}
\begin{observation}
\label{obs:points}
    There is a one-to-one correspondence between points in the verbose diagram
    recording a connected component and vertices of the graph, with birth
    coordinates equal to the heights of vertices. 
    That is, vertices are \emph{creators}.
    Additionally, there is a one-to-one correspondence between points in the
    verbose diagram recording cycles and points recording connected components with finite lifespans, with cycle births and component deaths equal to heights of edges. 
    That is, edges are either creators of cycles or \emph{destructors} of connected components.
\end{observation}
\newpage 
We now define the verbose persistent homology transform, specific to graphs:

\begin{definition}[Verbose Persistent Homology Transform]
Given a graph $G$ with vertices in~$\mathbb{R}^2$, the verbose
persistent homology transform (VPHT) of $G$, written $\mathit{VPHT}(G)$, is the
set of verbose diagrams corresponding to filtrations in all directions, parameterized by direction.
\end{definition}

In subsequent sections, we  discuss graphs with multi-edges, which are
themselves not simplicial complexes. We also discuss directed graphs. However,
we use multi- and directed edges only as a tool to understand pairs of simple
undirected graphs through their disjoint unions, and we only consider the VPHT
of undirected simple graphs.

Finally, we come to reconstructibility.
\begin{definition}[VPHT-Reconstructible]
   We say a graph $G_1$ with an embedded vertex set is \emph{VPHT reconstructible} if $\mathit{VPHT}(G_1) =
   \mathit{VPHT}(G_2)$ implies $G_1=G_2$. That is, the VPHT of $G_1$ is unique, and could be used to reconstruct $G_1$. Otherwise, $G_1$ is \emph{non
   VPHT-reconstructible}.
\end{definition}

Graphs in general position are always VPHT-reconstructible~\cite[Secs 4 and
5]{belton_reconstructing_2020}. We therefore focus on the simplest graphs breaking general position, graphs whose vertices are collinear and, without loss of generality, aligned with the $y$-axis.

\begin{figure}
    \centering
    \begin{tikzpicture}
    \draw[thick, fill =cyan!40] (0,0) ellipse (4cm and 2.5cm);
    
    \draw[thick,fill=blue!40] (-1,0) ellipse (3cm and 1.9cm);

    \draw[thick,fill=purple!40] (-1.8,0) ellipse (2.2cm and 1.3cm);
    
    \draw[thick,fill=red!50] (-3,0) ellipse (1cm and 0.7cm);
    
    \node at (3,0.2) {\small vertical};
    \node at (3,-0.2) {\small{graphs}};
    \node at (1.1,0.2) {\small colliding};
    \node at (1.1,-0.2) {\small pairs};
    \node at (-0.9,0.35) {\small vertical};
    \node at (-0.9,0.0) {\small non VPHT-}; \node at (-0.9,-0.3) {\small reconstructible};
    \node at (-3,0) {\small \cref{prop:special-colliding-pairs}};
\end{tikzpicture}
    \caption{The above Venn diagram shows the classification of different graphs that will be discussed in the following pages. All inclusions are strict, where \cref{lem:anothercycle} distinguishes the pink class from the red and \cref{fig:multiple edges} distinguishes the purple class from the pink.}
    \label{fig:venndiagramm}
\end{figure}
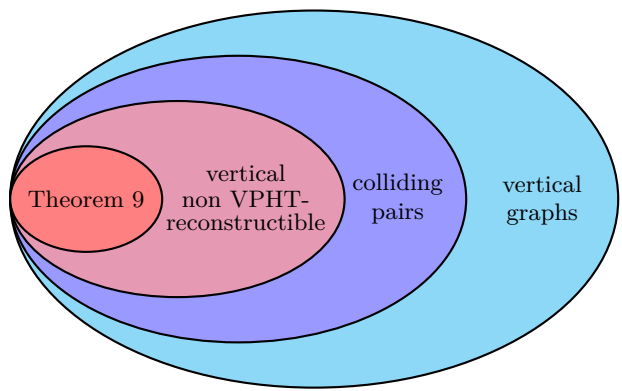

\begin{definition}[Vertical Graph]
    A \emph{vertical graph} $G$ is a graph whose vertices $(v_1,\ldots,v_n)$ are distinct points on a vertical line, such that $i<j $ implies that $v_j$ is above $v_i$, which we may also write as $v_i < v_j$. We write the height of $v_i$ in the up direction as~$i$.
\end{definition}

Note that in this setting, every direction encounters the vertices in one of
three orders; bottom-up, top-down, or simultaneously. The diagrams corresponding
to directions that see the vertices in one of these orders are related.
Directions that see vertices in the same total order produce nearly identical
persistence diagrams, differing only by the exact heights of events, but
\emph{not} in the order or type of events. Diagrams for the two directions that
see vertices simultaneously shift the order of events to the extreme, with every
event born at the same height and dying either instantaneously or
at infinity. These observations are more formally stated below, and allow us to
greatly simplify our proofs.

\begin{observation}
    Let $G$ be a vertical graph. 
 Then we can obtain $\mathit{VPHT}(G)$ from the verbose diagrams corresponding to only two directions, e.g.,
up and down along the vertical axis, as they can be used to obtain the verbose diagrams for any other direction. 
\end{observation}

To understand relationships between pairs of graphs $G_1$ and $G_2$, we often use
their \emph{disjoint union}, $G_1 \sqcup G_2$, by which we mean the ordinary
union of their vertices and disjoint union of their edges
(possibly leading to multi-edges). We utilize this idea in the
following definition.

\begin{definition}[Simple Colliding Pair and Alternating Cycle]
\label{def:collidingpair}
Vertical graphs $G_1$ and $G_2$ form a \emph{simple colliding pair} if orienting all edges of $G_1$ up and all edges of $G_2$ down creates a cycle in $G = G_1 \sqcup G_2$ that alternates between up and down edges, i.e., creates an \emph{alternating cycle}.
    See \cref{fig:firstEx}. 
\end{definition}

    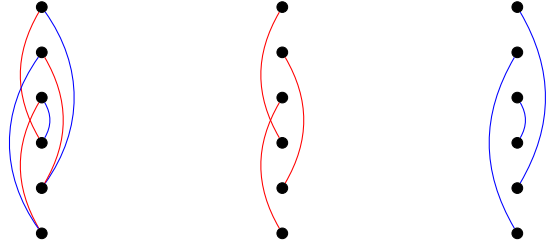
\begin{figure}[h]
        \centering
        \begin{tikzpicture}[node distance= 6mm]
            \node [D](a){};
            \node [above of=a,D](b){};
            \node [above of=b,D](c){};
            \node [above of=c,D](d){};
            \node [above of=d,D](e){};
            \node [above of=e,D](f){};
            \draw[bend left=30, color=red](a)to(d);
            \draw[bend left=30, color=blue](d)to(c);
            \draw[bend left=30, color=red](c)to(f);
            \draw[bend left=35, color=blue](f)to(b);
            \draw[bend right=30, color=red](b)to(e);
            \draw[bend right=35, color=blue](e)to(a);
        \end{tikzpicture}
        \hspace{2cm}
        \begin{tikzpicture}[node distance= 6mm]
            \node [D](a){};
            \node [above of=a,D](b){};
            \node [above of=b,D](c){};
            \node [above of=c,D](d){};
            \node [above of=d,D](e){};
            \node [above of=e,D](f){};
            \draw[bend left=30, color=red](a)to(d);
            \draw[bend left=30, color=red](c)to(f);
            \draw[bend right=30, color=red](b)to(e);
        \end{tikzpicture}
        \hspace{2cm}
        \begin{tikzpicture}[node distance= 6mm]
            \node [D](a){};
            \node [above of=a,D](b){};
            \node [above of=b,D](c){};
            \node [above of=c,D](d){};
            \node [above of=d,D](e){};
            \node [above of=e,D](f){};
            \draw[bend left=30, color=blue](d)to(c);
            \draw[bend left=30, color=blue](f)to(b);
            \draw[bend right=30, color=blue](e)to(a);
        \end{tikzpicture}
        \caption{The disjoint union (left) of a simple colliding pair (middle and right).
        We observe the alternating cycle if we consider all red edges to be oriented up (without loss of generality) and all blue edges to be oriented down.}
        \label{fig:firstEx}
    \end{figure}
    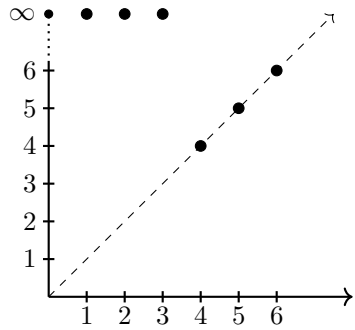
\begin{figure}[h]
        \centering
        \begin{tikzpicture}[scale=0.5]
            \draw[thick](0,0)--(0,6.2);
            \draw[thick,dotted](0,6.2)--(0,7.5);
            \node[S] at (0,7.5){};
            \draw[thick,->](0,0)--(8,0);
            \draw[->, dashed](0,0)--(7.5,7.5);
            \node at (-0.5,1){$1$};
            \node at (-0.5,2){$2$};
            \node at (-0.5,3){$3$};
            \node at (-0.5,4){$4$};
            \node at (-0.5,5){$5$};
            \node at (-0.5,6){$6$};
            \node at (1,-0.5){$1$};
            \node at (2,-0.5){$2$};
            \node at (3,-0.5){$3$};
            \node at (4,-0.5){$4$};
            \node at (5,-0.5){$5$};
            \node at (6,-0.5){$6$};
            \draw[thick] (-0.15,1)--(0.15,1);
            \draw[thick] (-0.15,2)--(0.15,2);
            \draw[thick] (-0.15,3)--(0.15,3);
            \draw[thick] (-0.15,4)--(0.15,4);
            \draw[thick] (-0.15,5)--(0.15,5);
            \draw[thick] (-0.15,6)--(0.15,6);
            \draw[thick] (1,-0.15)--(1,0.15);
            \draw[thick] (2,-0.15)--(2,0.15);
            \draw[thick] (3,-0.15)--(3,0.15);
            \draw[thick] (4,-0.15)--(4,0.15);
            \draw[thick] (5,-0.15)--(5,0.15);
            \draw[thick] (6,-0.15)--(6,0.15);
            \node at (-0.7,7.5){$\infty$};
            \node [D] at (1,7.5){};
            \node [D] at (2,7.5){};
            \node [D] at (3,7.5){};
            \node[D] at (4,4){};
            \node[D] at (5,5){};
            \node[D] at (6,6){};
        \end{tikzpicture}
        \caption{Verbose diagram for the red and blue graphs of \cref{fig:firstEx} when filtering in the up direction.}
        \label{fig:firstVPHT}
    \end{figure}
 \begin{definition}
    A graph of \emph{type $\mathcal{G}$} is a vertical graph, where the edges can be partitioned into alternating cycles $G^1,\ldots,G^k$.
    Furthermore, each $G^i$ can each be partitioned into a simple colliding pair $(G_1^{i},G_2^{i})$. The two subgraphs on the edges $\bigsqcup_{i\in [k]}G_1^{i}$, resp. $\bigsqcup_{i\in [k]}G_2^{i},$ are called graphs of \emph{type $\mathcal{G}_1$}, resp. $\mathcal{G}_2$. The pair $(\mathcal{G}_1,\mathcal{G}_2)$ is a \emph{colliding pair}.
\end{definition}
    Note that the choice of cycle partitioning $\mathcal{G}$ into some $\mathcal{G}_1 $ and $\mathcal{G}_2$ is not necessarily unique.
\begin{definition}
    For a vertical graph $G=(V,E)$, a fixed up direction, and vertex $v \in V$, we define the number of \emph{down-edges of $v$} as $\ldeg_G(v):=|\{w\in V\mid (v,w)\in E, w<v\}|$.
    Similarly, the number of \emph{up-edges of $v$} is $\hdeg_G(v):=|\{w\in V \mid (v,w)\in E, v<w\}|$.
\end{definition}

\section{Theoretical Results for Vertical Graphs}
\label{sec:theory}
Next, we classify families of vertical graphs that are non VPHT-reconstructible, summarized in \cref{fig:venndiagramm}. 
    We first observe that some colliding pairs are VPHT-reconstructible, e.g.,
    the graphs in \cref{fig:multiple edges}.   In addition to having distinct
    verbose diagrams in the up direction (\cref{fig:VPHT multiple edges}), we
    computationally checked that no other
    graph has the same VPHT. The following gives a family of
    colliding pairs that are \emph{not} reconstructible.

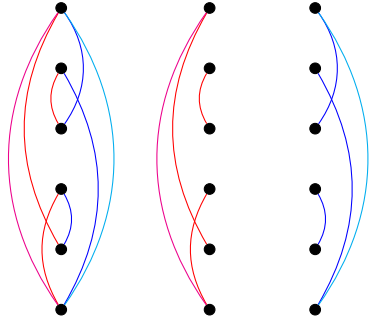
\begin{figure}[h]
    \centering
    \begin{tikzpicture}[node distance= 8mm]
            \node [D](a){};
            \node [above of=a,D](b){};
            \node [above of=b,D](c){};
            \node [above of=c,D](d){};
            \node [above of=d,D](e){};
            \node [above of=e,D](f){};
            \draw[bend left=30, color=red](a)to(c);
            \draw[bend left=30, color=blue](c)to(b);
            \draw[bend left=30, color=red](b)to(f);
            \draw[bend left=35, color=blue](f)to(d);
            \draw[bend left=30, color=red](d)to(e);
            \draw[bend left=30, color=blue](e)to(a);
            \draw[bend left=35, color=magenta] (a)to(f);
            \draw[bend left=35, color=cyan] (f)to(a);
        \end{tikzpicture}
        \begin{tikzpicture}[node distance= 8mm]
            \node [D](a){};
            \node [above of=a,D](b){};
            \node [above of=b,D](c){};
            \node [above of=c,D](d){};
            \node [above of=d,D](e){};
            \node [above of=e,D](f){};
            \draw[bend left=30, color=red](a)to(c);
            \draw[bend left=30, color=red](b)to(f);
            \draw[bend left=30, color=red](d)to(e);
            \draw[bend left=35, color=magenta] (a)to(f);
        \end{tikzpicture}
        \hspace{1cm}
        \begin{tikzpicture}[node distance= 8mm]
            \node [D](a){};
            \node [above of=a,D](b){};
            \node [above of=b,D](c){};
            \node [above of=c,D](d){};
            \node [above of=d,D](e){};
            \node [above of=e,D](f){};
            \draw[bend left=30, color=blue](c)to(b);
            \draw[bend left=35, color=blue](f)to(d);
            \draw[bend left=30, color=blue](e)to(a);
            \draw[bend left=35, color=cyan] (f)to(a);
        \end{tikzpicture}
    \caption{On the left $\mathcal{G}$, made up of two vertical graphs with alternating cycles; one includes all red and blue edges, the other the cyan and magenta edge. On the right the corresponding $\mathcal{G}_1$ and $\mathcal{G}_2$.}
    \label{fig:multiple edges}
\end{figure}

\begin{figure}[h]
        \centering
        \begin{minipage}{0.45\textwidth}
        \centering
        \begin{tikzpicture}[scale=0.5]
            \draw[thick](0,0)--(0,6.2);
            \draw[thick,dotted](0,6.2)--(0,7.5);
            \node[S] at (0,7.5){};
            \draw[thick,->](0,0)--(8,0);
            \draw[->, dashed](0,0)--(7.5,7.5);
            \node at (-0.5,1){$1$};
            \node at (-0.5,2){$2$};
            \node at (-0.5,3){$3$};
            \node at (-0.5,4){$4$};
            \node at (-0.5,5){$5$};
            \node at (-0.5,6){$6$};
            \node at (1,-0.5){$1$};
            \node at (2,-0.5){$2$};
            \node at (3,-0.5){$3$};
            \node at (4,-0.5){$4$};
            \node at (5,-0.5){$5$};
            \node at (6,-0.5){$6$};
            \draw[thick] (-0.15,1)--(0.15,1);
            \draw[thick] (-0.15,2)--(0.15,2);
            \draw[thick] (-0.15,3)--(0.15,3);
            \draw[thick] (-0.15,4)--(0.15,4);
            \draw[thick] (-0.15,5)--(0.15,5);
            \draw[thick] (-0.15,6)--(0.15,6);
            \draw[thick] (1,-0.15)--(1,0.15);
            \draw[thick] (2,-0.15)--(2,0.15);
            \draw[thick] (3,-0.15)--(3,0.15);
            \draw[thick] (4,-0.15)--(4,0.15);
            \draw[thick] (5,-0.15)--(5,0.15);
            \draw[thick] (6,-0.15)--(6,0.15);
            \node at (-0.7,7.5){$\infty$};
            \node [D] at (1,7.5){};
            \node [R] at (2,6){};
            \node [D] at (3,3){};
            \node[R] at (4,7.5){};
            \node[D] at (5,5){};
            \node[D] at (6,6){};
        \end{tikzpicture}
        \end{minipage}
        \begin{minipage}{0.45\textwidth}
        \centering
        \begin{tikzpicture}[scale=0.5]
            \draw[thick](0,0)--(0,6.2);
            \draw[thick,dotted](0,6.2)--(0,7.5);
            \node[S] at (0,7.5){};
            \draw[thick,->](0,0)--(8,0);
            \draw[->, dashed](0,0)--(7.5,7.5);
            \node at (-0.5,1){$1$};
            \node at (-0.5,2){$2$};
            \node at (-0.5,3){$3$};
            \node at (-0.5,4){$4$};
            \node at (-0.5,5){$5$};
            \node at (-0.5,6){$6$};
            \node at (1,-0.5){$1$};
            \node at (2,-0.5){$2$};
            \node at (3,-0.5){$3$};
            \node at (4,-0.5){$4$};
            \node at (5,-0.5){$5$};
            \node at (6,-0.5){$6$};
            \draw[thick] (-0.15,1)--(0.15,1);
            \draw[thick] (-0.15,2)--(0.15,2);
            \draw[thick] (-0.15,3)--(0.15,3);
            \draw[thick] (-0.15,4)--(0.15,4);
            \draw[thick] (-0.15,5)--(0.15,5);
            \draw[thick] (-0.15,6)--(0.15,6);
            \draw[thick] (1,-0.15)--(1,0.15);
            \draw[thick] (2,-0.15)--(2,0.15);
            \draw[thick] (3,-0.15)--(3,0.15);
            \draw[thick] (4,-0.15)--(4,0.15);
            \draw[thick] (5,-0.15)--(5,0.15);
            \draw[thick] (6,-0.15)--(6,0.15);
            \node at (-0.7,7.5){$\infty$};
            \node [D] at (1,7.5){};
            \node [B] at (2,7.5){};
            \node [D] at (3,3){};
            \node [B] at (4,6){};
            \node[D] at (5,5){};
            \node[D] at (6,6){};
        \end{tikzpicture}
        \end{minipage}
        \caption{Verbose persistence diagrams corresponding to $\mathcal{G}_1$ and $\mathcal{G}_2$ of \cref{fig:multiple edges} and the up direction; black dots appear in both diagrams, red only in $\mathcal{G}_1$, blue ones only in $\mathcal{G}_2$.}
        \label{fig:VPHT multiple edges}
    \end{figure}
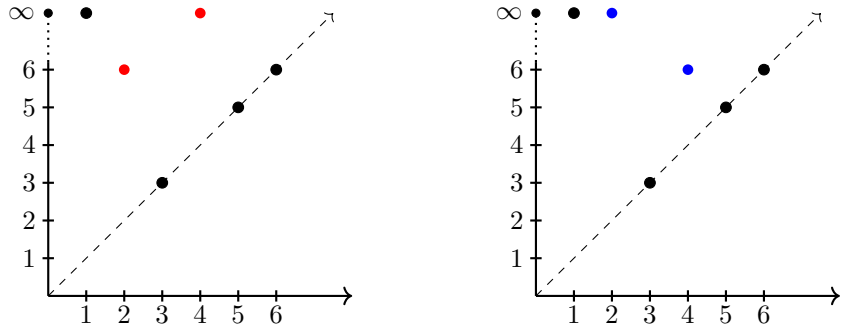
\begin{restatable}[Special Colliding Pairs are Non Reconstructible]{theorem}{SpecialCollidingPairsAreNon} 
    \label{prop:special-colliding-pairs}
    Consider a type $\mathcal{G}$ graph, such that each vertex has no more than two lower neighbors and no more than two upper neighbors. Then, for any choice of alternating cycles $C^1,\ldots,C^k$ partitioning $G$, the resulting graphs $\mathcal{G}_1$ and $\mathcal{G}_2$ are non VPHT-reconstructible.
\end{restatable}
Such a colliding pair can be seen in \cref{fig:firstEx} and an associated verbose diagram in \cref{fig:firstVPHT}.
\begin{proof}
    Fix some arbitrary choice of cycles.
    We focus on the up direction. Edges connecting a vertex to a vertex above will be called up-edges and those connecting a vertex to a vertex below will be called down-edges.
    
    Consider any vertex $v_i$ in $\mathcal{G}$. Then $v_i$ can have degree
    $\deg(v_i)\in\{0,2,4\}$, and $\mathcal{G}_{1}$ and~$\mathcal{G}_{2}$ each
    include half of the up-edges and half of the down-edges. To see this,
    assume $v_i$ has two lower neighbors in $\mathcal{G}$, denoted $v_j$ and $v_k$. Without
    loss of generality, assume $(v_j, v_i) \in \mathcal{G}_1$. So in our choice
    of cycle $C$, $v_j$ is followed by $v_i$. But then we also have that $v_i$
    must be followed by $v_k$, as an edge leading up must be followed by an edge
    leading down in an alternating cycle and thus the edge $(v_i,v_k)\in
    \mathcal{G}_2$.
    
    Next we show that $\mathcal{G}_1$ and $\mathcal{G}_2$ only have non-trivial zero-dimensional homology and thus their higher dimensional diagrams are trivially identical. Since we are considering graphs, the only higher order homology that could appear would be a cycle. 
    Assume there is such a cycle $v_{1},v_2,\ldots,v_{k}$ in $\mathcal{G}_{1}$. Let $v_{1}$ be its highest vertex. 
    Every vertex in the cycle must be connected to two other vertices in $\mathcal{G}_1$. But then $v_{1}$ must be connected to two lower vertices in $\mathcal{G}_{1}$; this is a contradiction to what we found in the previous paragraph. 
    
    So, to obtain a difference in the diagrams, there must be a difference in connected components at the height of some $v_i$. 
    The connected components in $\mathcal{G}_1$ and $\mathcal{G}_2$ are paths
    along vertices of increasing height or equivalently of increasing index,
    this follows from the second paragraph. 
    Then every such connected component is a singleton or has one minimal vertex
    $v_{\min}$, one maximal vertex $v_{\max}$ and some vertices in between. 
    In $\mathcal{G}$, $v_{\min}$ has degree two with two up-edges, hence one
    up-edge in each of $\mathcal{G}_{i}$ and $\mathcal{G}_2$. Similarly, $v_{max}$ has one
    down-edge in each of $\mathcal{G}_{i}$ and $\mathcal{G}_2$, and all other vertices in the
    connected component have one down and one up edge in $\mathcal{G}_1$ and $\mathcal{G}_2$
    Next, consider the degree of $v_i$.
    
    Case 1: deg($v_i$)=0, then $v_i=v_{\min}=v_{\max}$. The vertex $v_i$ is a singleton, thus there is be a connected component born at height $i$ that lives until infinity, corresponding to a point~$(i,\infty)$ in both diagrams.
    
    Case 2: deg($v_i)=2$, then a)~$v_i=v_{min}\neq v_{\max}$ or b)~$v_i=v_{max}\neq v_{\min}$ in both $\mathcal{G}_1$ and~$\mathcal{G}_2$. In case 2a, a new connected component is born at height $i$ and both diagrams have a point at~$(i,\infty)$. Case 2b leads to a point $(i,i)$ on the diagonal, since the vertex created a new connected component that is instantly killed by adding the edge to its preceding vertex. 
    
    Case 3: deg($v_i$)=4, then $v_i$ is a middle vertex in some connected components in $\mathcal{G}_1$ and~$\mathcal{G}_2$, so we again obtain a point at $(i,i)$.

    Hence, for all heights $i$, we get the same points in both diagrams and thus the diagrams of the two graphs coincide, and they are non VPHT-reconstructible.
\end{proof}

\commentHack{\begin{proof}
    See \cref{proof:special-col-non-recon}.
\end{proof}
}

Next, we identify properties of non-reconstructible graphs. In order to prove \cref{thm:VPHTimplCollision}, we first establish the following fact about degrees of vertices in type $\mathcal{G}$ graphs.

\commentHack{
\begin{proof}
    See \cref{proof:odd-degress}.
\end{proof}
}

\begin{restatable}[Type $\mathcal{G}$ Equivalence]{restLemma}{TypeGEquivalence}
    \label{lemma:even-equiv}
    A vertical graph $G=(V,E)$ is of type $\mathcal{G}$ if and only if for all $v \in V$, $\mathrm{ldeg}_G(v)$ and $ \mathrm{hdeg}_G(v)$ are even. 
\end{restatable}

\begin{proof}
If $G$ is of type $\mathcal{G}$ then it can be partitioned into alternating cycles. As an alternating cycle visits every vertex incoming from the same direction as it exits (that is, either from above or below), it is evident that every vertex has an even $\mathrm{ldeg}$ and $\mathrm{hdeg}$.

For the reverse, we prove the contrapositive. 
Assume $G$ is not of type $\mathcal{G}$, and thus cannot be partitioned into alternating cycles. 
If it cannot be partitioned into cycles then by Veblen's theorem~\cite{veblen1912application}, there exists a vertex of odd degree and we are done. 

Otherwise, we assume $G$ can be partitioned into cycles. We proceed to construct a partition into cycles through recursion, for which we can guarantee that there exists a vertex violating the even-ness property. 
By assumption, there exists an initial partition into cycles $P_0=(C^0_1, \dots, C^0_n)$. 
For the recursion step, assume $P_i=(C_1^i,\dots,C_{n-i}^i)$ is a partition of the edges into cycles. 
By assumption, there exists a non-alternating cycle $C$ within it. 
As such, without loss of generality, this cycle visits a vertex $v$ coming from below and exiting above.

If every other cycle in $P_i$ visiting $v$ enters $v$ from the same direction as it exits, we have found a vertex $v$ with odd $\mathrm{hdeg}$ and $\mathrm{ldeg}$, fulfilling the desired property. 

Otherwise, again up to reversing cycle direction, there exists another cycle $C'$ entering~$v$ from above and exiting below. 
In this case, we construct $P_{i+1}$ by joining the two cycles into a single cycle $\widetilde{C}$, 
such that the (into $v$) incoming edge of $C$ is succeeded by the outgoing edge of $C'$ and the incoming edge of $C'$ is succeeded by the outgoing edge of $C$ in the new cycle~$\widetilde{C}$. 
This new cycle visits $v$ twice, both times exiting in the same direction it enters. 

After at most $n-1$ steps of this recursive construction, we find a vertex $v$ with $\mathrm{hdeg}$ and $\mathrm{ldeg}$ being odd in a partition along the way, or are left with only a single cycle, in which case by the assumption that it cannot be alternating (as otherwise the partition would be alternating), it must contain such a vertex.

We have shown if $G$ is not of type $\mathcal{G}$ there exists some $v \in V$ such that $\ldeg_G(v)$ or $\hdeg_G(v)$ is odd. Thus, we have shown the desired result through contraposition.
\end{proof}

This allows us to state the following theorem about non-reconstructibility.

\commentHack{\begin{proof}
    See \cref{proof:type-g-equiv}.
\end{proof}}

\begin{theorem}[Non-Reconstructible implies Colliding Pair]
    \label{thm:VPHTimplCollision}
    If a pair of vertical graphs~$(G_1,G_2)$ is not a colliding pair, then $\mathit{VPHT}(G_1)$ and $\mathit{VPHT}(G_2)$ are distinct.
\end{theorem}

\commentHack{\begin{proof}
    See \cref{proof:non-recon-impl-coll}.
\end{proof}}

\begin{proof}
We consider the verbose diagrams corresponding to the up direction.
If $G_1$ and $G_2$ are not a colliding pair,
    then, by \cref{lemma:even-equiv}, there exists some vertex~$v_i$ that
    (without loss of generality) $\ldeg(v_i)$ is odd and $\ldeg_{G_1}(v_i) >
    \ldeg_{G_2}(v_i)$.

    By \cref{obs:points}, edges incident to and below $v_i$ in $G_1$ are in bijection with points in the verbose diagram marking zero-dimensional deaths at height $i$ and one-dimensional births at height $i$. Then, because the number of such edges in $G_1$ and $G_2$ differ, the number of such points in the diagrams for $G_1$ and $G_2$ must also differ. Thus, $\mathit{VPHT}(G_1) \neq \mathit{VPHT}(G_2)$.
\end{proof}
Finally, given a colliding pair, we can extend it so that it remains (non-) reconstructible.
\begin{restatable}{restLemma}{finalLemmaOfSubm}
    Let $(\mathcal{G}_1,\mathcal{G}_2)$ be a colliding pair. Let $\mathcal{C}=\{C^1,\dots,C^k\}$ be a partition of $G$ into alternating cycles. Then adding another copy of a cycle $C\in \mathcal{C}$ to $\mathcal{G}$ does not change the non VPHT-reconstructibility of the resulting colliding pair.
    \label{lem:anothercycle}
\end{restatable}

\begin{proof}
We investigate how adding a cycle copy impacts the verbose diagrams with respect to the up direction.
Adding a copy of an existing cycle cannot connect previously disconnected components, meaning there is no change to the zero-dimensional homology. 
    We obtain a new cycle for every new double edge, born at the height of the
    higher vertex of the edge. But since we added an alternating cycle, a vertex
    $v$ is the higher vertex for an edge in $\mathcal{G}_1\sqcup C$ if and only if it is
    the higher vertex for an edge in $\mathcal{G}_2\sqcup C$, 
    since that means we have a sequence~$v_i,v,v_j$ in the cycle, with $v_i,v_j<v$, thus the edge $(v_i,v)$ must be in $\mathcal{G}_1\sqcup C$ and $(v,v_j)$ in $\mathcal{G}_2\sqcup C$.
    As this argument applies to both copies of the added cycle, double edges with upper vertex~$v$ appear in both of the new graphs in the new pair (obtained after partitioning the new $G$ with added cycle $C$), thus the same one-cycles are born at the same heights, and thus the changes to both verbose diagrams are identical. An identical argument holds for the down direction. Thus, if the VPHTs were identical (or distinct), they remain identical (or distinct), maintaining the (non-) reconstructibility of the resulting graphs.
\end{proof}

\commentHack{\begin{proof}
    See \cref{proof:adding-copies}.
\end{proof}}

\section{Computational Approach}\label{overview}
To complement the theoretical exploration of vertical graphs and alternating cycles, we developed a computational tool with the purpose of more easily investigating the relationship between topological features of certain graphs and their VPHT. This approach was motivated by the hypothesis that an exhaustive search on small graphs might reveal structural patterns that obstruct unique reconstructibility, as well the question whether all non VPHT-reconstructible pairs exhibit the previously detailed cycle pattern. 

\paragraph{Computational Limits}
In addition to restricting ourselves to finite simple graphs, we would like to highlight some limitations of our computations. For one, any brute force search of the kind we perform very quickly runs into a major scaling problem. The number of possible graphs that need to be considered on any given $n$ vertices is $2^{\left(\frac{n(n-1)}{2}\right)}$. For us, computations on six vertices are near instant, on seven it takes a few seconds, for eight we could be looking at half an hour and anything beyond that is out of reach. These scaling laws apply no matter how powerful the hardware.

Our particular implementation has a hard limit of allowing at most 32 vertices in search algorithms, however this is practically irrelevant as the number of graphs to search there is on the order of $10^{149}$.

\subsection{Application Feature Overview}\label{app-overview}
We developed a graphical user interface (GUI) based app facilitating responsive exploration. It is accessible as a static page progressive web app or as a native app compiled from source. In both cases it runs entirely locally. Both are publicly accessible:

\begin{itemize}
    \item Progressive web app: \href{https://kalokak.github.io/vpht-algo/}{https://kalokak.github.io/vpht-algo/}
    \item Public source code: \href{https://github.com/KalokaK/vpht-algo}{https://github.com/KalokaK/vpht-algo}
\end{itemize}

The application allows users to draw a graph and visualize its persistence diagrams for zeroth and first dimensional homology in a given selected \emph{sweep direction}. The application can then perform a brute force search for sets of graphs with the same VPHT. Note that they compute diagrams only in the selected sweep direction and its reverse for a given graph $G$; this approach is equivalent to considering the VPHT of the graph where all vertices of $G$ are projected on a line, and thus aligns with our theoretical approach. There are various different settings for the brute force search and options on how to display the output, discussed in \cref{subsec:features}. The algorithms used in the application are listed in \cref{subsec:algos}. With the help of this tool, we established the results given in the following section.

\subsection{Numerical Results}
Results were obtained using the "Compute Collision Sets" feature, with the corresponding number of vertically stacked vertices drawn, "Ignore Dangling Vertices" disabled and no edges drawn.

\begin{lemma}[Partitionability up to 7 Vertices]
\label{lem:part7}
    All pairs of vertical graphs with up to seven vertices and that have the same VPHT are colliding pairs.
\end{lemma}

\begin{remark}
\label{rem:about-part7}
    \Cref{lem:part7} was proven computationally with "Common Edges Excluded from Cycle Search" (see \cref{subsec:features}) as edges common to both always from a cycle, and thus their inclusion cannot break partitionability.
\end{remark}

\begin{corollary}[Independence up to 7 Vertices]
    \label{cor:independent7}
    All pairs of vertical graphs with up to seven vertices that have the same VPHT form colliding pairs, and there exists an alternating cycle partition such that all common edges form cycles of length two. 
\end{corollary}
\begin{proof}
    By \cref{rem:about-part7} and \cref{lem:part7}, we conclude that the disjoint union graph for all such pairs with up to seven vertices has the desired property.
\end{proof}
\begin{remark}
    This implies, removing all common edges, we still obtain a decomposition into alternating cycles.
We conjecture that an analogue to \cref{cor:independent7} holds for an arbitrary number of vertices.
\end{remark}

\section{Discussion and Conclusion}
\label{sec:discussion}
We presented a partial classification of non VPHT-reconstructible vertical graphs (\cref{fig:venndiagramm}).
We are still interested in a complete classification in terms of easily checkable conditions.

In our work, we only consider vertical graphs.
For general graphs, one can reconstruct any edge whose supporting line contains no vertices other than its two endpoints \cite{belton_reconstructing_2020}.
We conjecture that if a particular edge of a graph $G$ is not reconstructible, then the set of edges of $G$ with the same supporting line themselves form a (rotated) vertical graph that is non VPHT-reconstructible.
If true, this conjecture implies that any non VPHT-reconstructible graph decomposes into reconstructible edges, and non VPHT-reconstructible vertical graphs.

While theoretical work on reconstructibility often involves carefully chosen
directions, current applications of directional transforms typically use random or
evenly-spaced directions, the effect of which was
explored in \cite{fasy2025fewsamplingboundstopological}. Our results give insight to directions that are potentially unhelpful to choose in such settings. Namely,
if projecting a graph onto some one-dimensional subspace produces one of
our non VPHT-reconstructible vertical graphs, we know a direction orthogonal to
such a subspace is not able to uniquely reconstruct the edges. For instance, this is exactly
the case for the graphs in \cite[Figure 8]{ordering}, where the authors identify
all directions in which the graphs project to a version of
Figure~\ref{fig:firstEx}. 

\bibliography{eurogc_citations.bib}

\newpage
\appendix

\markboth{Appendix}{Appendix}
\crefalias{section}{appendix} 
\crefalias{subsection}{appendix}
\crefname{appendix}{Appendix}{Appendices}
\Crefname{appendix}{Appendix}{Appendices}

\section{Computational Approach}\label{appendix:computation}
\subsection{Additional features of the GUI}\label{subsec:features}

There are two main types of brute force searches: 
\begin{enumerate}
    \item Find all possible graphs on the same set of vertices as the drawn graph $G$, compute their VPHT, and record all graphs $G'$ with the same VPHT as $G$. They are depicted as \emph{colliding pairs} with $G$. This feature is called `{Compute Colliding Graphs}.'
    \item Find the set of all possible graphs $\Omega$ on the same set of vertices as the drawn graph, which include the (possibly empty) set of edges of the drawn graph. Within $\Omega$ search for sets of graphs which all share the same VPHT, so-called \emph{collision sets}\footnote{Notice how the "Compute Colliding Graphs" brute force search is nothing but a search for the collision set which contains the drawn graph.}.
    The application then allows users to filter and sort these collision sets by number of connected components in the graphs\footnote{\label{foot:constant}As a result of having the same VPHT both of these are constant within a collision set.}, number of cycles \footref{foot:constant}, number of off diagonal points in the diagram, longest cycle found within a set\footnote{How this is defined and works is document in the UI, and further ellaborated on in \cref{finding-minimal-partitions}.} and if the set contains a colliding pair which does not partition into cycles. This feature is called `{Compute Colliding Sets}.'
\end{enumerate}
For both of these brute force search methods, we implemented a cycle search available through the `{Show Cycles}' button. This will attempt to partition the graphs in a collision set pairwise into cycles, if no partition is found, non exists.

Finally, there are two main configuration options for the brute force search:
\begin{enumerate}
    \item `{Ignore Dangling Vertices':} Determines if graphs with isolated vertices should be ignored for the search. Note that this may include the original graph for "Compute Colliding Graphs" in which case no results will be returned.
    \item `{Exclude Common Edges from Cycle Search':} Determines  whether to disregard edges common to both graphs for the cycle search in colliding pairs. 
\end{enumerate}
The tool features various tooltips accessible by {hovering over the various features,} and a more in-depth explanation of the search and sort functionality through a "help" button.

\subsection{Algorithms}\label{subsec:algos}
Here, we briefly describe the central algorithms used for the computational tool, which are relevant to the numerical results we obtained. For performance reasons, the actual implementation is highly parallel and optimizes certain parts of the procedures, but remains computationally equivalent to what we outline here. 

\subsubsection{Computing the VPHT}\label{computing-persistence}
To obtain the diagrams, we follow much of the same steps computationally as we do mathematically in \cref{prelim}. For a given sweep direction (and its reverse) we start by computing a filtration. This is done using a standard sweepline algorithm approach, as described in \cref{algo:sweepline}.
\begin{algorithm}
\caption{Sweep Line Filtration}\label{algo:sweepline}
\begin{algorithmic}[1]
\Require{Graph $G=(V,E)$, Sweep Direction $\vec{d}$}
\Ensure{Vertices $V_{ord}$ sorted by height, Edges $E_{ord}$ sorted by birth height.}
\Function{BuildFiltration}{$G, \vec{d}$}
    \State $V_{ord} \gets$ List of $(v, \text{height}(v, \vec{d}))$ for all $v \in V$
    \State Sort $V_{ord}$ ascending by height
    \State $E_{ord} \gets \emptyset$

    \For{$i \gets |V_{ord}| -1$ \textbf{down to} $1$}
    \Comment{Iterate from highest to lowest}
        \State $(u, h_u) \gets V_{ord}[i]$
        \Comment{Check against all strictly lower vertices}
        \For{$j \gets 0$ \textbf{to} $i-1$}
            \State $(v, \_) \gets V_{ord}[j]$
            \If{edge $(u,v)$ exists in $G$}
                \State $e_{new} \gets (i,j)$ \Comment{Edges store indices to $V_{ord}$}
                \State Append $(e_{new}, h_u)$ to $E_{ord}$ \Comment{Birth height is height of upper vertex}
            \EndIf
        \EndFor
    \EndFor

    \State Reverse $E_{ord}$ \Comment{Resulting order: increasing brith height}
    \State \Return $(V_{ord}, E_{ord})$
\EndFunction
\end{algorithmic}
\end{algorithm}
We then use the matrix reduction algorithm to find creator destructor pairs, from which we compute the verbose persistence diagrams. To reduce the complexity of comparing persistence diagrams further along in the process, we store the points sorted lexicographically first by birth then by death height. 
As we assume collinear vertices (in practice graphs projected onto the sweep direction),
the diagrams for the sweep direction and its reverse contain the same information as the entire VPHT.

\subsubsection{Computing Colliding
Graphs}\label{computing-colliding-graphs}
Given that we can compute the VPHT for the subset of graphs we treat in our work checking for colliding pairs on the set of vertices given by the graph is a simple task. We enumerate all possible edge combinations\footnote{\label{foot:respect}Respecting the option for "Ignore Dangling Vertices".} and compare the VPHT of the resulting graph to the VPHT of the drawn graph. We collect all graphs that form a colliding pair with the drawn graph, including the graph itself, and display the result to the user. 
The user has the option to search for cycles for colliding pairs in the result.

\subsubsection{Computing Collision Sets}\label{computing-collision-sets}
Finding all possible graphs $\Omega$ with the same vertex set as the drawn graph $G$ that contain $G$ as a subgraph
 remains computationally simple. However, the partition of those graphs into colliding sets is more involved. Ultimately we are seeking to obtain the equivalence classes of the relation "identical VPHT" in $\Omega$ and any algorithm achieving this will lead to an equivalent implementation. Nonetheless, we briefly present our approach here.

Our approach is based on the idea that one can hash data and use a truncation of the hash to index into an array, to obtain $\mathcal{O}(\lvert \Omega \rvert)$ quotienting, at the cost of memory. 
We derive a hash for the VPHT based on the FNV1A
hash algorithm~\cite{fowler_fnv_nodate}.
We construct an array of "lists-of-collision-sets" with a power of $2$ size which is $\mathcal{O}(\lvert \Omega \rvert)$ and for each graph index into it using the appropriate number of bits from the hash. 
The resulting list-of-collision-sets is then searched for a set which matches the graph's VPHT, if none are found a new one is added to the list. The graph is inserted into the appropriate set. 
At the end, we collect all non-singleton collision sets.

\subsubsection{Searching for Alternating Cycles and Finding Minimal
Partitions}\label{searching-for-up-down-cycles-and-finding-minimal-partitions}
We have mentioned cycle search numerous times throughout the previous sections. In effect, we are interested in determining if a pair of graphs that have the same VPHT are a simple colliding pair (Definition \ref{def:collidingpair}). Here we outline an algorithmic approach to this.

Given two graphs $G_1$ and $G_2$ on the same vertex set, we construct their disjoint union graph and try to split it into alternating cycles. We label all edges in $G_1$ red up-edges and all edges in $G_2$ blue down-edges. 
The algorithm we used to find partitions of the disjoint union graph into alternating cycles consists of two layers of Depth First Search (DFS); the first layer to find cycles and the second layer to find possible partitions into cycles.

  \begin{enumerate}

  \item In the first layer, using DFS, we find all alternating
    cycles, which start at a red up-edge with lower vertex $v$ and terminate with a blue down-edge with lower vertex $v$.
  \item In the second, in order to find partitions we first pick an arbitrary up-edge $e_0$. Find all cycles $C_0$ starting at $e_0$ using DFS. 
  For each such cycle $c_0 \in C_0$ consider the remaining edges not appearing in $c_0$ from $G_1, G_2$. Out of these, pick an arbitrary red up-edge $e_1$ and again find all cycles $C_1$ within the remaining edges not used in any previous cycles starting at $e_1$. Recurse, until no red up-edges are left remaining.
    \begin{enumerate}
    \item
      At this point, if all blue down-edges have also been used, we have found a partition into cycles, and we append it to a list of possible partitions.
    \item
      If blue down-edges are left remaining, we backtrack, discarding the last cycle $c_k$ and continue recursion by replacing it with the next cycle from $C_k$.
      If no cycles are left in~$C_k$ we again backtrack, discarding $c_{k-1}$ and consider the next cycle from $C_{k-1}$.
    \end{enumerate}
  \item
    We have found all possible partitions.
  \end{enumerate}

\paragraph{\textbf{Finding Minimal
Partitions.}}\label{finding-minimal-partitions}

If two cycles share a vertex, we can combine them into a single, longer cycle. However, this longer cycle is of little interest to us, as it is really just two shorter cycles glued together. If we are for example interested in finding graphs that feature length eight cycles in their partition, we are not interested in length eight cycles which can be broken down into two length four cycles, rather, we want to find cycles of length eight that cannot be split apart. 

Thus, we call a partition minimal if none of its cycles can be further subdivided into shorter cycles. In our algorithm, we expand this further: If we have two partitions into cycles
$P_1$ and $P_2$, we can consider the lengths of cycles in $P_1$ and $P_2$ in descending order and represent $P_1$ as $(l_{1}, l_{2},\dots,l_{k_{1}}, 0, 0, \dots)$ with $l_{i}\geq l_{i+1}$ and  $P_2$ as $(m_{1}, m_{2},\dots,m_{k_{2}}, 0, 0, \dots)$ with $m_i \geq m_{i+1}$. We then compare these in lexicographical order, which we denote $\triangleleft$. We pick a partition into cycles that is minimal with respect to this ordering to represent the colliding pair. Further work could explore all possible partitions into
cycles and look for interesting results there; this was not our focus.

Lastly, we would like to note that we perform this "minimum finding" in-place during the search and not as a secondary step, as it is vastly more memory efficient. An algorithmic description is depicted in \cref{ppair}.

\begin{algorithm}
\caption{Partition Pair}\label{ppair}
\begin{algorithmic}[1]
\Require{Graphs \((V, E_a), (V,E_b)\) on the same vertices.}
\Ensure{Lexicographically minimal partition $P_{min}$}
\State \(p_{min} \gets \left(\infty\right)\) \Comment{Current minimal partition.}
\State \(p \gets \emptyset\) \Comment{Current partition.}
\State \(U_a \gets \emptyset\) \Comment{Used edges.}
\State \(U_b \gets \emptyset\) \Comment{Used edges.}

\Function{Search}{$p, U_a, U_b$}
    \If{$U_a = E_a \wedge U_b = E_b$}
        \If{$p \triangleleft  p_{min}$} \Comment{"$\triangleleft$" denotes lexicographic ordering as above.}
            \State $p_{min} \gets p$
        \EndIf
        \State \Return
    \EndIf

    \State Pick $e \in E_a \setminus U_a$
    \State $C \gets \Call{CycleDFS}{e, U_a, U_b}$
    \For{$c \in C$}
        \State Add $c$ edges to $U_a, U_b$
        \State Call \Call{Search}{$p \cup \{c\}, U_a, U_b$}
        \State Remove $c$ edges from $U_a, U_b$
    \EndFor
\EndFunction
\Function{CycleDFS}{$e, U_a, U_b$}
\State Perform DFS recursively, alternating between unused edges from $E_a$ and $E_b$ starting with $e$. 
\State \Return List of found cycles.
\EndFunction
\end{algorithmic}
\end{algorithm}

\end{document}